\documentclass{amsart}

\usepackage[T1]{fontenc}
\usepackage[utf8]{inputenc}
\usepackage{amsmath,amssymb,amsthm,mathtools}
\usepackage{hyperref}
\hypersetup{hidelinks}

\newtheorem{theorem}{Theorem}[section]
\newtheorem{proposition}[theorem]{Proposition}
\newtheorem{lemma}[theorem]{Lemma}
\newtheorem{corollary}[theorem]{Corollary}

\theoremstyle{definition}
\newtheorem{example}[theorem]{Example}
\newtheorem{remark}[theorem]{Remark}

\DeclareMathOperator{\spn}{span}
\newcommand{\N}{\mathbb N}
\newcommand{\C}{\mathbb C}
\newcommand{\fall}[2]{(#1)_{#2}}

\begin{document}

\title[Finite-term recurrences in a Bochner--Krall family]
{Finite-term recurrences in a generalized Bochner--Krall family}

\author[L.~M.~Anguas]{L.~M.~Anguas}
\address{Saint Louis University, Madrid Campus, Avenida del Valle 34, 28003 Madrid, Spain}
\email{luismiguel.anguas@slu.edu}

\author[D.~Barrios Rolan\'ia]{D.~Barrios Rolan\'ia}
\address{ETSI Industriales, Universidad Polit\'ecnica de Madrid,
C/Jos\'e Guti\'errez Abascal 2, 28006 Madrid, Spain}
\email{dolores.barrios.rolania@upm.es}

\author[B.~Shapiro]{B.~Shapiro}
\address{Department of Mathematics, Stockholm University, SE-106 91, Stockholm, Sweden}
\email{shapiro@math.su.se}

\author[M.~Tater]{M.~Tater}
\address{Department of Theoretical Physics, Nuclear Physics Institute,
Academy of Sciences, 250\,68 \v{R}e\v{z} near Prague, Czech Republic}
\email{tater@ujf.cas.cz}

\subjclass[2020]{Primary 42C05, 34A30; Secondary 33C45, 47E05}
\keywords{generalized Bochner--Krall problem, exactly solvable operator,
polynomial eigenfunction, finite-term recurrence, $d$-orthogonality}

\begin{abstract}
We classify the differential operators
\(T=z^j\partial_z^j+z^m\partial_z^\ell\), where
\(0\le m<\ell\) and \(1\le j<\ell\), whose monic eigenpolynomials satisfy a
finite-term recurrence relation.
Writing \(k=\ell-m\), such a recurrence exists if and only if \(j=1\) and
\(k\mid\ell\).  In that case we determine all recurrence coefficients in
closed form and prove that the associated difference operator has order
\(\ell\).  We also give an explicit factorization showing that every
admissible operator is of Type~(2) in Conjecture~1.10 of
Horozov--Shapiro--Tater; the equality of the differential and difference
orders is the conclusion predicted by their Conjecture~1.9.
\end{abstract}

\maketitle

\section{Introduction}

The algebraic higher-order Bochner--Krall problem asks which differential
operators possess polynomial eigenfunctions that also satisfy a recurrence
relation with a fixed number of terms.  A general framework and several
conjectures were proposed in \cite{HoShTa}; see also \cite{Be} for the
asymptotic theory of polynomial eigenfunctions of exactly solvable
differential operators, and \cite{Segundo} for recurrence methods based on
the quantities \(\delta_n^{(r)}\).

In this paper we consider the three-parameter family
\begin{equation}\label{eq:T}
T=T_{j,m,\ell}:=z^j\partial_z^j+z^m\partial_z^\ell,
\qquad 0\le m<\ell,\quad 1\le j<\ell,
\end{equation}
and we set
\[
k:=\ell-m\in\{1,\dots,\ell\}.
\]

Since \(T\) maps \(z^n\) to \(\fall{n}{j}z^n+\fall{n}{\ell}z^{n-k}\), it is
exactly solvable: it preserves the degree filtration of \(\C[z]\) and is
triangular in the monomial basis.  For \(n\ge j\) there is a unique monic
eigenpolynomial \(P_n\) of degree \(n\) (Lemma~\ref{lem:eigenpolys}), while
for \(n<j\) every monic polynomial of degree \(n\) is an eigenpolynomial,
with eigenvalue \(0\).  We remove the latter ambiguity by imposing the
normalization
\begin{equation}\label{eq:normalization}
P_n(z)=z^n,\qquad 0\le n<j.
\end{equation}
This costs nothing; see Remark~\ref{rem:normalization}.

Since \(\{P_n\}_{n\ge0}\) is a basis of \(\C[z]\) and \(\deg(zP_n)=n+1\),
there are unique scalars \(\alpha_{n,R}\) with
\begin{equation}\label{eq:alpha-def}
zP_n=\sum_{R=0}^{n+1}\alpha_{n,R}P_R,
\qquad \alpha_{n,n+1}=1 .
\end{equation}
We say that \(\{P_n\}_{n\ge0}\) \emph{satisfies a recurrence relation with a
fixed number of terms}, or that it has \emph{finite lower bandwidth}, if
there is an integer \(N\) such that
\begin{equation}\label{eq:general-rec}
\alpha_{n,R}=0\quad\text{whenever } n-R\ge N,
\end{equation}
that is, if
\(
zP_n=P_{n+1}+\sum_{s=0}^{N-1}\alpha_{n,n-s}P_{n-s}
\)
for every \(n\), with the convention \(P_r:=0\) for \(r<0\).  The smallest
such \(N\) is called the optimal bandwidth.

Our main result is the following complete classification.

\begin{theorem}\label{thm:main}
Let \(T\) be given by \eqref{eq:T}, let \(k=\ell-m\), and let \(\{P_n\}\) be
the normalized monic eigenpolynomials of \(T\).  Then \(\{P_n\}\) satisfies a
recurrence relation with a fixed number of terms if and only if
\[
j=1
\qquad\text{and}\qquad
k\mid\ell .
\]
Assume these two conditions, and put \(\sigma:=\ell/k\).  Then
\begin{equation}\label{eq:main-rec}
zP_n=P_{n+1}+\sum_{q=1}^{\sigma}\gamma_n^{(q)}P_{n+1-qk},
\end{equation}
where
\begin{equation}\label{eq:main-gamma}
\gamma_n^{(q)}
=
\frac{(-1)^q}{q!\,k^q}
\left(\prod_{a=0}^{q-1}(\ell-ak)\right)
\left(\prod_{c=0}^{q-1}\fall{n-ck}{\ell-1}\right),
\end{equation}
and \(\fall{x}{r}:=x(x-1)\cdots(x-r+1)\).  The associated difference
operator has order \(\ell\); equivalently, the optimal lower bandwidth is
\(N=\ell\), independently of \(m\).
\end{theorem}

The two cases \(m=0\) and \(m=\ell-1\) correspond to the two extreme divisors
\(k=\ell\) and \(k=1\); further admissible values of \(m\) occur precisely
when \(\ell\) is composite (Corollary~\ref{cor:prime}).  The smallest new
case is \((\ell,m,j)=(4,2,1)\), that is
\(T=z\partial_z+z^2\partial_z^4\); see Example~\ref{ex:counterexample}.

We also wish to point out a trap.  For \(k\ge2\) the sequence \(\{P_n\}\) is
\((k-1)\)-symmetric in the sense of Douak and Maroni, i.e.
\(P_n(wz)=w^n P_n(z)\) for \(w=e^{2\pi i/k}\).  By
\cite[Th\'eor\`eme 5.1]{Douak}, a \((k-1)\)-symmetric and
\((k-1)\)-orthogonal sequence obeys a recurrence of the special shape
\(P_{n+k}=zP_{n+k-1}-\gamma_n P_n\), with a single lower term.  It is tempting
to conclude that a symmetric sequence with \emph{some} fixed-term recurrence
must obey a recurrence of that shape.  This is not so: the cited theorem
presupposes \((k-1)\)-orthogonality, which is precisely the statement that
the recurrence has \(k+1\) terms.  Symmetry alone yields only the congruence
of Proposition~\ref{prop:congruence}, and Example~\ref{ex:counterexample}
exhibits a \(1\)-symmetric sequence with two lower diagonals.  Accordingly,
the proofs below control all permitted lower diagonals simultaneously.

The paper is organized as follows.  Section~\ref{sec:eigen} collects the
eigenpolynomials, their support, and the congruence constraint.
Section~\ref{sec:bidiag} sets up a bidiagonal model on the \(T\)-invariant
subspaces spanned by a single residue class of monomials; this is where all
the content lies, and it is uniform in \(m\).  Section~\ref{sec:necessity}
proves that \(j\ge2\) is impossible, and Section~\ref{sec:jone} treats
\(j=1\) by an explicit evaluation, completing the proof of
Theorem~\ref{thm:main}.  We conclude by identifying the admissible operators
with the Type~(2) family in Conjecture~1.10 of \cite{HoShTa}.

\section{Eigenpolynomials, support, and the congruence constraint}
\label{sec:eigen}

Throughout, \(\fall{x}{r}=x(x-1)\cdots(x-r+1)\) and \(\fall{x}{0}=1\).  On
monomials,
\begin{equation}\label{eq:Tmon}
Tz^n=\fall{n}{j}z^n+\fall{n}{\ell}z^{n-k},
\end{equation}
so that the eigenvalue attached to \(P_n\) is
\begin{equation}\label{eq:eigenvalue}
\lambda_n=\fall{n}{j}.
\end{equation}
Note that \(\lambda_x=\fall{x}{j}\) vanishes for \(0\le x<j\) and is strictly
increasing for \(x\ge j\); hence
\begin{equation}\label{eq:distinct}
\lambda_n\ne\lambda_s
\qquad\text{whenever } n\ge j \text{ and } 0\le s<n .
\end{equation}

\begin{lemma}\label{lem:eigenpolys}
Let \(n\ge j\) and write \(P_n(z)=\sum_{s=0}^n b_{n,s}z^s\) with
\(b_{n,n}=1\).  Then \(P_n\) exists and is unique, and its coefficients are
determined by
\begin{equation}\label{eq:coefficient-rec}
\bigl(\lambda_n-\lambda_s\bigr)b_{n,s}
=
\fall{s+k}{\ell}\,b_{n,s+k},
\qquad 0\le s\le n,
\end{equation}
where \(b_{n,r}:=0\) for \(r>n\).  Moreover:
\begin{enumerate}
\item \(b_{n,s}=0\) unless \(s\equiv n \pmod k\), and \(b_{n,s}=0\) for
\(s<m\);
\item \(P_n=z^n\) for \(j\le n<\ell\);
\item for \(n\ge\ell\),
\begin{equation}\label{eq:eigenpoly-support}
P_n(z)=\sum_{t=0}^{r_n}b_{n,n-tk}\,z^{n-tk},
\qquad r_n:=\left\lfloor\frac{n-m}{k}\right\rfloor,
\end{equation}
with
\begin{equation}\label{eq:eigenpoly-coeff}
b_{n,n-tk}
=
\prod_{i=1}^{t}
\frac{\fall{n-(i-1)k}{\ell}}{\fall{n}{j}-\fall{n-ik}{j}},
\qquad 0\le t\le r_n .
\end{equation}
\end{enumerate}
If \(j=1\), then the right-hand side of \eqref{eq:eigenpoly-coeff} equals
\(0\) for every \(t>r_n\) as well, so that \eqref{eq:eigenpoly-coeff} is
valid for all \(t\ge0\).
\end{lemma}

\begin{proof}
Comparing the coefficients of \(z^s\) on both sides of
\(TP_n=\lambda_nP_n\) and using \eqref{eq:Tmon} gives
\eqref{eq:coefficient-rec}.  By \eqref{eq:distinct}, for \(n\ge j\) the
factor \(\lambda_n-\lambda_s\) is nonzero for every \(s<n\); hence
\eqref{eq:coefficient-rec} determines \(b_{n,s}\) from \(b_{n,s+k}\) for all
\(s<n\), which proves existence and uniqueness.

If \(s+k>n\), the right-hand side of \eqref{eq:coefficient-rec} vanishes, so
\(b_{n,s}=0\); descending induction on \(s\) in steps of \(k\) shows that
\(b_{n,s}=0\) unless \(s\equiv n\pmod k\).  If \(s<m\), then
\(s+k<\ell\) and therefore \(\fall{s+k}{\ell}=0\), whence \(b_{n,s}=0\);
this proves (1).  If \(j\le n<\ell\), then \(\fall{n}{\ell}=0\), so
\eqref{eq:coefficient-rec} with \(s=n-k\) gives \(b_{n,n-k}=0\), and
descending induction gives (2).

For \(n\ge\ell\), iterating \eqref{eq:coefficient-rec} downwards from
\(s=n\) yields \eqref{eq:eigenpoly-coeff}; the range of \(t\) is bounded by
(1), because \(n-tk\ge m\) is equivalent to \(t\le r_n\).  Finally,
\(n-r_n k\) lies in \(\{m,\dots,m+k-1\}=\{m,\dots,\ell-1\}\), so
\(\fall{n-r_n k}{\ell}=0\); thus for \(t>r_n\) the numerator in
\eqref{eq:eigenpoly-coeff} contains the factor \(\fall{n-r_n k}{\ell}=0\).
When \(j=1\), the denominator equals
\(\prod_{i=1}^{t}\bigl(n-(n-ik)\bigr)=k^t\,t!\ne0\), so the whole expression
vanishes, as claimed.  (For \(j\ge2\) we simply keep the convention
\(b_{n,s}=0\) outside the range described in (1) and (3).)
\end{proof}

For \(0\le c\) we write
\begin{equation}\label{eq:Vc}
\mathcal V_c:=\spn\{z^{c+ik}:i\ge0\}\subset\C[z].
\end{equation}
Then
\begin{equation}\label{eq:decomposition}
\C[z]=\bigoplus_{r=0}^{k-1}\mathcal V_r ,
\end{equation}
and, by Lemma~\ref{lem:eigenpolys} together with the normalization
\eqref{eq:normalization},
\begin{equation}\label{eq:P-in-V}
P_n\in\mathcal V_{r}\qquad\text{for } r\equiv n \pmod k,\ 0\le r<k .
\end{equation}
Equivalently, \(P_n(wz)=w^n P_n(z)\) with \(w=e^{2\pi i/k}\), i.e.\ the
sequence \(\{P_n\}\) is \((k-1)\)-symmetric.

\begin{proposition}\label{prop:congruence}
For all \(n\) and \(R\),
\[
\alpha_{n,R}=0
\qquad\text{unless}\qquad
R\equiv n+1 \pmod k .
\]
Consequently
\begin{equation}\label{eq:symmetric-rec}
zP_n=P_{n+1}+\sum_{q\ge1}\gamma_n^{(q)}P_{n+1-qk},
\qquad \gamma_n^{(q)}:=\alpha_{n,n+1-qk},
\end{equation}
where for each fixed \(n\) the sum is finite.  The sequence \(\{P_n\}\) has
finite lower bandwidth if and only if the number of nonzero summands in
\eqref{eq:symmetric-rec} is bounded uniformly in \(n\).
\end{proposition}

\begin{proof}
Fix \(0\le r<k\).  By \eqref{eq:P-in-V} the polynomials \(P_n\) with
\(n\equiv r \pmod k\) lie in \(\mathcal V_r\), have pairwise distinct
degrees, and one of each degree in \(r+k\N\) occurs; hence they form a basis
of \(\mathcal V_r\).  Now \(P_n\in\mathcal V_{n\bmod k}\) implies
\(zP_n\in\mathcal V_{(n+1)\bmod k}\).  Expanding \(zP_n\) in the basis
\(\{P_R\}_{R\ge0}\) of \(\C[z]\) and comparing with its expansion in the
basis of \(\mathcal V_{(n+1)\bmod k}\), the uniqueness of coordinates with
respect to the direct sum \eqref{eq:decomposition} forces
\(\alpha_{n,R}=0\) for \(R\not\equiv n+1 \pmod k\).  The last assertion is
immediate from \eqref{eq:general-rec}.
\end{proof}

We stress that Proposition~\ref{prop:congruence} is unconditional: no
recurrence is assumed.  It is also all that symmetry gives.

\begin{example}\label{ex:counterexample}
Let \((\ell,m,j)=(4,2,1)\), that is
\[
T=z\partial_z+z^2\partial_z^4,
\qquad k=2 .
\]
Here \(\lambda_n=n\) and \(P_n=z^n\) for \(n\le3\), while
\begin{align*}
P_4&=z^4+12z^2, & P_5&=z^5+60z^3,\\
P_6&=z^6+180z^4+1080z^2,
& P_7&=z^7+420z^5+12600z^3.
\end{align*}
and so on.  For every \(n\ge0\) one has
\begin{equation}\label{eq:counterexample}
zP_n
=
P_{n+1}
-2n(n-1)(n-2)\,P_{n-1}
+n(n-1)(n-2)^2(n-3)(n-4)\,P_{n-3} .
\end{equation}
Thus the sequence is \(1\)-symmetric and has finite lower bandwidth
\(N=4=\ell\), but \emph{two} lower diagonals occur.  Since the coefficient of
\(P_{n-3}\) is eventually nonzero, the sequence is \(3\)-orthogonal in the
standard terminology, not \(1\)-orthogonal.  Formula
\eqref{eq:counterexample} is the case \(k=\sigma=2\) of
Theorem~\ref{thm:main}.
\end{example}

\begin{remark}\label{rem:normalization}
The normalization \eqref{eq:normalization} is harmless.  If \(j=1\) there is
no freedom at all, since \(P_0=1\) is forced; and \(j=1\) is the only case in
which a fixed-term recurrence exists.  If \(j\ge2\), the proof of
Theorem~\ref{thm:jge2} uses only the polynomials \(P_n\) with
\(n\ge\ell-1\ge j\), which are unique; hence the negative conclusion holds
for every choice of \(P_0,\dots,P_{j-1}\).
\end{remark}

\section{A bidiagonal model on a residue class}\label{sec:bidiag}

Fix an integer \(c\) with \(0\le c\le\ell-1\) and consider the subspace
\(\mathcal V_c\) of \eqref{eq:Vc}.  Put
\[
e^{(c)}_i:=z^{c+ik},
\qquad
\Lambda^{(c)}_i:=\lambda_{c+ik}=\fall{c+ik}{j},
\qquad
M^{(c)}_i:=\fall{c+ik}{\ell}.
\]
We drop the superscript \((c)\) when no confusion is possible.

\begin{lemma}\label{lem:basis}
Let \(0\le c\le\ell-1\).  Then:
\begin{enumerate}
\item \(M_0=0\), and \eqref{eq:Tmon} reads
\begin{equation}\label{eq:bidiagonal}
Te_0=\Lambda_0e_0,
\qquad
Te_i=\Lambda_ie_i+M_ie_{i-1}\quad(i\ge1);
\end{equation}
in particular \(T\mathcal V_c\subseteq\mathcal V_c\).  Also
\(z^k\mathcal V_c\subseteq\mathcal V_c\), with \(z^k e_i=e_{i+1}\).
\item \(P_{c+ik}\in e_i+\spn\{e_0,\dots,e_{i-1}\}\) for every \(i\ge0\);
consequently \(\{P_{c+ik}\}_{i\ge0}\) is a basis of \(\mathcal V_c\).
\item If moreover \(c\ge m\), then \(M_i>0\) for every \(i\ge1\).
\end{enumerate}
\end{lemma}

\begin{proof}
(1) \(M_0=\fall{c}{\ell}=0\) because \(0\le c\le\ell-1\); the rest is
\eqref{eq:Tmon}.

(2) If \(c+ik<\ell\) then \(P_{c+ik}=z^{c+ik}=e_i\) by
Lemma~\ref{lem:eigenpolys}(2) and \eqref{eq:normalization}.  If
\(c+ik\ge\ell\), then by Lemma~\ref{lem:eigenpolys}(1),(3) all exponents
occurring in \(P_{c+ik}\) are congruent to \(c\) modulo \(k\) and are at
least \(m\); the smallest one lies in \(\{m,\dots,m+k-1\}\).  If \(c\ge m\),
that smallest exponent is \(c\) itself, because \(c\) is the unique element
of \(\{m,\dots,m+k-1\}=\{m,\dots,\ell-1\}\) congruent to \(c\).  If
\(c<m\), the smallest exponent is \(c+ak\) for the unique positive integer
\(a\) which places it in this interval.  In either case
\(P_{c+ik}\in\mathcal V_c\), and the leading term is \(e_i\).  Triangularity gives the basis statement.

(3) \(M_i=\fall{c+ik}{\ell}\) is nonzero as soon as \(c+ik\ge\ell\); and if
\(c\ge m\) and \(i\ge1\) then \(c+ik\ge m+k=\ell\).
\end{proof}

Part (3) may fail for \(c<m\): for instance \(M_1=\fall{c+k}{\ell}=0\)
whenever \(c<m\).  This is harmless below, since
Lemma~\ref{lem:coordinate} does not require \(M_i\ne0\).

Let \(\psi^{(c)}_R\in\mathcal V_c^{*}\) denote the coordinate functionals of
this basis,
\[
\psi^{(c)}_R\bigl(P_{c+ik}\bigr)=\delta_{Ri}.
\]

\begin{lemma}\label{lem:coordinate}
Let \(0\le c\le\ell-1\) and assume that the numbers \(\Lambda_i\)
\((i\ge0)\) are pairwise distinct; this holds automatically if \(c\ge j\),
and also if \(j=1\).  Then \(\psi^{(c)}_R\circ T=\Lambda_R\psi^{(c)}_R\) and
\begin{equation}\label{eq:coordinate}
\psi^{(c)}_R(e_i)=
\begin{cases}
0,&i<R,\\[2pt]
1,&i=R,\\[2pt]
\displaystyle(-1)^{i-R}\prod_{q=R+1}^{i}
\dfrac{M_q}{\Lambda_q-\Lambda_R},&i>R.
\end{cases}
\end{equation}
\end{lemma}

\begin{proof}
First, the \(\Lambda_i\) are pairwise distinct under either hypothesis: by
\eqref{eq:distinct} the map \(x\mapsto\fall{x}{j}\) is injective on
\(\{x\ge j\}\) and vanishes on \(\{0\le x<j\}\), so it is injective on the
increasing sequence \(c+ik\) as soon as at most one term of that sequence is
smaller than \(j\); this is the case when \(c\ge j\) (no term) and when
\(j=1\) (at most the term \(c=0\)).

By Lemma~\ref{lem:basis}(2), \(T\) restricted to \(\mathcal V_c\) is diagonal
in the basis \(\{P_{c+ik}\}\), with \(TP_{c+ik}=\Lambda_iP_{c+ik}\); hence
\(\psi^{(c)}_R\circ T=\Lambda_R\psi^{(c)}_R\).  Applying this identity to
\(e_i\) and using \eqref{eq:bidiagonal} gives
\begin{align*}
\Lambda_i\psi^{(c)}_R(e_i)+M_i\psi^{(c)}_R(e_{i-1})
&=\Lambda_R\psi^{(c)}_R(e_i),\\
\psi^{(c)}_R(e_i)
&=-\frac{M_i}{\Lambda_i-\Lambda_R}\,\psi^{(c)}_R(e_{i-1}).
\end{align*}
for \(i\ne R\).  Together with \(\psi^{(c)}_R(e_i)=0\) for \(i<R\) and
\(\psi^{(c)}_R(e_R)=1\), which follow from the triangularity in
Lemma~\ref{lem:basis}(2), this gives \eqref{eq:coordinate}.
\end{proof}

\section{The necessity of \texorpdfstring{\(j=1\)}{j=1}}\label{sec:necessity}

\begin{theorem}\label{thm:jge2}
Let \(2\le j<\ell\) and \(0\le m<\ell\).  Then the monic eigenpolynomials of
\(T\) do not satisfy any recurrence relation with a fixed number of terms.
\end{theorem}

\begin{proof}
Throughout we work on \(\mathcal V_p\) with
\[
p:=\ell-1\ \ (\ge j),
\]
and we abbreviate \(e_i=z^{p+ik}\), \(\Lambda_i=\fall{p+ik}{j}\),
\(M_i=\fall{p+ik}{\ell}\) and \(\psi_R=\psi^{(p)}_R\).  By
Lemma~\ref{lem:basis}, \(M_0=0\) and \(M_i>0\) for \(i\ge1\) (indeed
\(p+ik\ge\ell-1+k\ge\ell\)), and \(\{P_{p+ik}\}_{i\ge0}\) is a basis of
\(\mathcal V_p\).  Since \(p\ge j\), the sequence \((\Lambda_i)_{i\ge0}\) is
strictly increasing, so Lemma~\ref{lem:coordinate} applies; in particular
\(\Lambda_i-\Lambda_0>0\) for \(i\ge1\).

Suppose, for contradiction, that \eqref{eq:general-rec} holds for some
\(N\ge1\).  Multiplication by \(z\) then lowers the index \(R\) by at most
\(N-1\), so multiplication by \(z^k\) lowers it by at most \(k(N-1)\).
Since \(z^k P_{p+ik}\in\mathcal V_p\) and the basis of \(\mathcal V_p\) is a
subset of the basis \(\{P_n\}\) of \(\C[z]\), the coefficient of \(P_p\) in
\(z^k P_{p+ik}\), which is \(\chi(P_{p+ik})\) for
\[
\chi:=\psi_0\circ(\text{multiplication by }z^k)\in\mathcal V_p^{*},
\]
vanishes whenever
\[
p<(p+ik)-k(N-1),
\]
that is, whenever \(ik>k(N-1)\), and consequently whenever \(i\ge N\).
Both \(\chi\) and
\(\sum_{R=0}^{N-1}\chi(P_{p+Rk})\,\psi_R\) therefore take the same value on
every element of the basis \(\{P_{p+ik}\}_{i\ge0}\), whence
\[
\chi=\sum_{R=0}^{N-1}c_R\psi_R,
\qquad c_R:=\chi(P_{p+Rk}).
\]

Evaluate this identity at \(e_i\) with \(i\ge N\).  Since \(z^k e_i=e_{i+1}\),
Lemma~\ref{lem:coordinate} gives
\[
(-1)^{i+1}\prod_{q=1}^{i+1}\frac{M_q}{\Lambda_q-\Lambda_0}
=
\sum_{R=0}^{N-1}c_R(-1)^{i-R}
\prod_{q=R+1}^{i}\frac{M_q}{\Lambda_q-\Lambda_R}.
\]
Dividing by the nonzero quantity
\((-1)^i\prod_{q=1}^{i}M_q/(\Lambda_q-\Lambda_0)\), we obtain
\begin{equation}\label{eq:contradiction}
-\frac{M_{i+1}}{\Lambda_{i+1}-\Lambda_0}
=
\sum_{R=0}^{N-1}\widehat c_R\,\Pi_R(i)
\qquad (i\ge N),
\end{equation}
where the constants \(\widehat c_R\) and the products \(\Pi_R(i)\) are
\[
\widehat c_R
:=(-1)^R c_R\prod_{q=1}^{R}\frac{\Lambda_q-\Lambda_0}{M_q},
\qquad
\Pi_R(i):=\prod_{q=R+1}^{i}\frac{\Lambda_q-\Lambda_0}{\Lambda_q-\Lambda_R}.
\]
Note that \(\widehat c_R\) does not depend on \(i\).

The left-hand side of \eqref{eq:contradiction} tends to \(-\infty\).  Indeed
\(M_{i+1}=\fall{p+(i+1)k}{\ell}\sim(ki)^{\ell}\) and
\(\Lambda_{i+1}-\Lambda_0\sim(ki)^{j}\) as \(i\to\infty\), so the left-hand
side is asymptotic to \(-(ki)^{\ell-j}\), and \(\ell-j\ge1\) by hypothesis.

The right-hand side of \eqref{eq:contradiction}, on the contrary, stays
bounded.  Fix \(R\).  For \(q>R\) we have \(\Lambda_q>\Lambda_R\ge\Lambda_0\),
so every factor of \(\Pi_R(i)\) satisfies
\[
\frac{\Lambda_q-\Lambda_0}{\Lambda_q-\Lambda_R}
=1+\frac{\Lambda_R-\Lambda_0}{\Lambda_q-\Lambda_R}\ \ge\ 1,
\]
and \(\Pi_R(i)\) is nondecreasing in \(i\).  Moreover
\(\Lambda_q=\fall{p+qk}{j}\) is a polynomial in \(q\) of degree \(j\) with
leading coefficient \(k^j>0\), so there is \(C_R>0\) with
\(\Lambda_q-\Lambda_R\ge C_R q^{j}\) for all large \(q\).  Since \(j\ge2\),
\[
\sum_{q>R}\frac{\Lambda_R-\Lambda_0}{\Lambda_q-\Lambda_R}<\infty,
\]
and \(\log(1+x)\le x\) yields
\(\log\Pi_R(i)\le\text{const}<\infty\).  Hence each \(\Pi_R(i)\) increases to
a finite limit, and the right-hand side of \eqref{eq:contradiction} converges
to a finite number as \(i\to\infty\).  This contradiction proves the theorem.
\end{proof}

\begin{remark}\label{rem:boundary}
The hypothesis \(j\ge2\) is used only for the convergence of
\(\sum_q q^{-j}\).  For \(j=1\) one has \(\Lambda_q=p+qk\), so
\(\Pi_R(i)=\binom{i}{R}\) and the right-hand side of
\eqref{eq:contradiction} is a polynomial in \(i\) of degree at most
\(N-1\), whereas the left-hand side equals
\(-\prod_{a=1}^{\ell-1}\bigl((i+1)k+a\bigr)\), a polynomial of degree
\(\ell-1\).  There is then no contradiction; one even reads off that
\(N\ge\ell\), in accordance with Theorem~\ref{thm:main}.
\end{remark}

\begin{remark}
For \(m=\ell-1\) one has \(k=1\) and \(\mathcal V_p=z^{\ell-1}\C[z]\), and
\(\chi(P_{p+N})\) is the coefficient of \(P_{\ell-1}\) in
\(zP_{\ell-1+N}\).  The proof of Theorem~\ref{thm:jge2} shows that this
coefficient is nonzero for infinitely many \(N\) whenever \(1<j<\ell\).
\end{remark}

\section{The case \texorpdfstring{\(j=1\)}{j=1}}\label{sec:jone}

Throughout this section \(j=1\), so that \(\lambda_n=n\), and \(k=\ell-m\).
By Lemma~\ref{lem:eigenpolys},
\begin{equation}\label{eq:b-jone}
b_{n,n-tk}
=\frac{1}{k^t\,t!}\prod_{i=0}^{t-1}\fall{n-ik}{\ell}
\qquad\text{for all } t\ge0 .
\end{equation}

\begin{lemma}\label{lem:coordinate-jone}
Let \(R\ge0\), write \(R=r+sk\) with \(0\le r<k\) and \(s\ge0\), and let
\(\psi:=\psi^{(r)}_s\) be the coordinate functional of \(P_R\) in the basis
\(\{P_{r+ik}\}_{i\ge0}\) of \(\mathcal V_r\).  Then
\begin{equation}\label{eq:psi-jone}
\psi\bigl(z^{R+ik}\bigr)
=\frac{(-1)^i}{k^i\,i!}\prod_{a=1}^{i}\fall{R+ak}{\ell},
\qquad i\ge0 .
\end{equation}
\end{lemma}

\begin{proof}
Since \(0\le r<k\le\ell\), Lemmas~\ref{lem:basis} and \ref{lem:coordinate}
apply to \(\mathcal V_r\) (the hypothesis of Lemma~\ref{lem:coordinate} holds
because \(j=1\)).  With \(\Lambda_q=\lambda_{r+qk}=r+qk\) we get
\(\Lambda_q-\Lambda_s=(q-s)k\), and \eqref{eq:coordinate} gives, for
\(i\ge0\),
\[
\psi\bigl(z^{R+ik}\bigr)
=\psi\bigl(e^{(r)}_{s+i}\bigr)
=(-1)^{i}\prod_{q=s+1}^{s+i}\frac{\fall{r+qk}{\ell}}{(q-s)k}
=\frac{(-1)^i}{k^i\,i!}\prod_{a=1}^{i}\fall{R+ak}{\ell},
\]
where we substituted \(q=s+a\) and used \(r+(s+a)k=R+ak\).
\end{proof}

\begin{theorem}\label{thm:jone}
Let \(j=1\) and \(k=\ell-m\).  Then \(\alpha_{n,R}=0\) unless
\(k\mid(n+1-R)\), and for all \(\tau\ge1\) and all \(n\) with
\(R:=n+1-\tau k\ge0\),
\begin{equation}\label{eq:closed-alpha}
\alpha_{n,\,n+1-\tau k}
=
\frac{(-1)^\tau}{\tau!\,k^\tau}
\left(\prod_{a=0}^{\tau-1}(\ell-ak)\right)
\left(\prod_{c=0}^{\tau-1}\fall{n-ck}{\ell-1}\right).
\end{equation}
\end{theorem}

\begin{proof}
The congruence restriction is Proposition~\ref{prop:congruence}.

Fix \(\tau\ge1\) and let \(n\) satisfy \(R=n+1-\tau k\ge0\).  Since
\(zP_n=\sum_{t\ge0}b_{n,n-tk}z^{n+1-tk}\) and \(zP_n\in\mathcal V_r\) with
\(r\equiv R \pmod k\), applying the functional \(\psi\) of
Lemma~\ref{lem:coordinate-jone} gives
\begin{equation}\label{eq:alpha-as-sum}
\alpha_{n,R}
=\sum_{t=0}^{\tau}b_{n,n-tk}\,\psi\bigl(z^{n+1-tk}\bigr),
\end{equation}
the terms with \(t>\tau\) being zero because \(\psi\) annihilates
\(z^{R'}\) for \(R'<R\), \(R'\equiv R\).  By \eqref{eq:b-jone} and
\eqref{eq:psi-jone}, each summand of \eqref{eq:alpha-as-sum} is a polynomial
function of \(n\) (recall from Lemma~\ref{lem:eigenpolys} that
\eqref{eq:b-jone} is valid for every \(t\ge0\) when \(j=1\)).  Hence
\(n\mapsto\alpha_{n,n+1-\tau k}\) agrees with a polynomial on
\(\{n:n\ge\tau k-1\}\).  The right-hand side of \eqref{eq:closed-alpha} is a
polynomial in \(n\) as well, so it suffices to prove
\eqref{eq:closed-alpha} for all sufficiently large \(n\).  We may therefore
assume
\begin{equation}\label{eq:n-large}
n\ge m+\tau k,
\end{equation}
which guarantees \(n+1-ck\ge n+1-(\tau-1)k\ge\ell+1\) for \(0\le c\le\tau-1\);
in particular \(\fall{n+1-ck}{\ell}\ne0\) for those \(c\).

Substituting \eqref{eq:b-jone} and \eqref{eq:psi-jone} into
\eqref{eq:alpha-as-sum}, and using \(R+ak=n+1-(\tau-a)k\) to rewrite
\(\prod_{a=1}^{\tau-t}\fall{R+ak}{\ell}=\prod_{c=t}^{\tau-1}\fall{n+1-ck}{\ell}\),
we get
\[
\alpha_{n,n+1-\tau k}
=\sum_{t=0}^{\tau}
\frac{1}{k^t\,t!}\prod_{i=0}^{t-1}\fall{n-ik}{\ell}
\cdot
\frac{(-1)^{\tau-t}}{k^{\tau-t}(\tau-t)!}
\prod_{c=t}^{\tau-1}\fall{n+1-ck}{\ell}.
\]
Using \(\frac{1}{t!(\tau-t)!}=\frac{1}{\tau!}\binom{\tau}{t}\), pulling out
the (nonzero) factor \(\prod_{c=0}^{\tau-1}\fall{n+1-ck}{\ell}\), and
observing that
\[
\frac{\fall{n-ik}{\ell}}{\fall{n+1-ik}{\ell}}
=\frac{n+1-\ell-ik}{n+1-ik},
\]
we arrive at
\begin{equation}\label{eq:pre-vandermonde}
\alpha_{n,n+1-\tau k}
=
\frac{(-1)^\tau}{k^\tau\tau!}
\left(\prod_{c=0}^{\tau-1}\fall{n+1-ck}{\ell}\right)
\sum_{t=0}^{\tau}(-1)^t\binom{\tau}{t}
\prod_{i=0}^{t-1}\frac{n+1-\ell-ik}{n+1-ik}.
\end{equation}

Put
\[
w:=\frac{n+1}{k},
\qquad
\sigma:=\frac{\ell}{k}
\]
(here \(\sigma\) need not be an integer).  Writing
\(n+1-ik=k(w-i)\) and \(n+1-\ell-ik=k(w-\sigma-i)\), and using the rising
factorials \(x^{(t)}:=x(x+1)\cdots(x+t-1)\), we obtain
\[
\prod_{i=0}^{t-1}\frac{n+1-\ell-ik}{n+1-ik}
=\prod_{i=0}^{t-1}\frac{w-\sigma-i}{w-i}
=\frac{(\sigma-w)^{(t)}}{(-w)^{(t)}},
\qquad
(-1)^t\binom{\tau}{t}=\frac{(-\tau)^{(t)}}{t!},
\]
so that the finite sum in \eqref{eq:pre-vandermonde} equals the terminating
Gauss series
\[
\sum_{t=0}^{\tau}\frac{(-\tau)^{(t)}}{t!}\,
\frac{(\sigma-w)^{(t)}}{(-w)^{(t)}}
={}_2F_1\bigl(-\tau,\sigma-w;-w;1\bigr).
\]
Here the Gauss hypergeometric series terminates because its first upper
parameter is the nonpositive integer \(-\tau\); for this notation and the
Chu--Vandermonde summation used below, see, for example,
\cite[Appendix~III]{Slater}.  By \eqref{eq:n-large} we have \(w>\tau-1\), so
\((-w)^{(\tau)}\ne0\) and the Chu--Vandermonde identity applies:
\[
{}_2F_1\bigl(-\tau,\sigma-w;-w;1\bigr)
=\frac{\bigl((-w)-(\sigma-w)\bigr)^{(\tau)}}{(-w)^{(\tau)}}
=\frac{(-\sigma)^{(\tau)}}{(-w)^{(\tau)}}
=\frac{\prod_{a=0}^{\tau-1}(\sigma-a)}{\prod_{a=0}^{\tau-1}(w-a)} .
\]
Multiplying numerator and denominator by \(k^\tau\) turns the last fraction
into
\[
\frac{\prod_{a=0}^{\tau-1}(\ell-ak)}{\prod_{a=0}^{\tau-1}(n+1-ak)} .
\]
Substituting this into \eqref{eq:pre-vandermonde} and using
\(\fall{y}{\ell}/y=\fall{y-1}{\ell-1}\) with \(y=n+1-ck\) gives
\eqref{eq:closed-alpha}.
\end{proof}

\begin{corollary}\label{cor:classification}
For the operator \eqref{eq:T}, the sequence of normalized monic
eigenpolynomials satisfies a recurrence relation with a fixed number of
terms if and only if \(j=1\) and \(k=\ell-m\) divides \(\ell\).  In that
case \eqref{eq:main-rec} and \eqref{eq:main-gamma} hold with
\(\sigma=\ell/k\), and the optimal bandwidth is \(N=\ell\).
\end{corollary}

\begin{proof}
By Theorem~\ref{thm:jge2} we may assume \(j=1\) and use
Theorem~\ref{thm:jone}.  By Proposition~\ref{prop:congruence}, the only
possibly nonzero coefficients are the \(\alpha_{n,n+1-\tau k}\) with
\(\tau\ge1\), and these are given by \eqref{eq:closed-alpha}.  The factor
\begin{equation}\label{eq:the-factor}
\prod_{a=0}^{\tau-1}(\ell-ak)
\end{equation}
vanishes if and only if \(\ell=ak\) for some \(0\le a\le\tau-1\), that is,
if and only if \(k\mid\ell\) and \(\tau\ge\sigma+1\), where
\(\sigma=\ell/k\).  Furthermore
\(\prod_{c=0}^{\tau-1}\fall{n-ck}{\ell-1}\ne0\) as soon as
\(n-(\tau-1)k\ge\ell-1\).

If \(k\nmid\ell\), then \eqref{eq:the-factor} never vanishes, so for every
\(\tau\ge1\) we get \(\alpha_{n,n+1-\tau k}\ne0\) for all
\(n\ge\ell-1+(\tau-1)k\).  Since the index drop \(n-(n+1-\tau k)=\tau k-1\)
is unbounded, the lower bandwidth is infinite.

If \(k\mid\ell\), then \eqref{eq:the-factor} vanishes for every
\(\tau\ge\sigma+1\), so \(\alpha_{n,R}=0\) whenever
\(n-R\ge\sigma k=\ell\); this is \eqref{eq:general-rec} with \(N=\ell\), and
the surviving coefficients \(\gamma^{(q)}_n=\alpha_{n,n+1-qk}\),
\(1\le q\le\sigma\), are exactly \eqref{eq:main-gamma}.  Finally, for
\(\tau=\sigma\) one has
\(\prod_{a=0}^{\sigma-1}(\ell-ak)=k^\sigma\prod_{a=0}^{\sigma-1}(\sigma-a)
=k^\sigma\sigma!\), whence
\[
\gamma_n^{(\sigma)}
=(-1)^\sigma\prod_{c=0}^{\sigma-1}\fall{n-ck}{\ell-1}\ne0
\qquad\text{for } n\ge\ell-1+(\sigma-1)k=2\ell-k-1 .
\]
The corresponding index drop is \(\sigma k-1=\ell-1\), so \(N=\ell\) cannot
be lowered.
\end{proof}

Theorem~\ref{thm:main} is the conjunction of
Corollary~\ref{cor:classification} and Theorem~\ref{thm:jge2}.

\begin{corollary}\label{cor:prime}
Fix \(\ell\ge2\).  Then \(m=0\) and \(m=\ell-1\) are the only admissible
values of \(m\) in Theorem~\ref{thm:main} if and only if \(\ell\) is prime.
\end{corollary}

\begin{proof}
The admissible \(m\) are exactly the \(m=\ell-k\) with \(k\mid\ell\),
\(1\le k\le\ell\).  The divisors of \(\ell\) in \(\{1,\dots,\ell\}\) reduce
to \(\{1,\ell\}\) precisely when \(\ell\) is prime, and \(k=1\), \(k=\ell\)
correspond to \(m=\ell-1\), \(m=0\).
\end{proof}

\begin{remark}
The two extreme cases read as follows.  For \(m=0\) we have \(k=\ell\),
\(\sigma=1\), and \eqref{eq:main-rec} becomes the sparse relation
\[
zP_n=P_{n+1}-\fall{n}{\ell-1}\,P_{n+1-\ell}.
\]
For \(m=\ell-1\) we have \(k=1\) and \(\sigma=\ell\); since
\(\prod_{a=0}^{q-1}(\ell-a)=\ell!/(\ell-q)!\), formula
\eqref{eq:main-gamma} with \(q=s+1\) reads
\[
\alpha_{n,n-s}
=(-1)^{s+1}\binom{\ell}{s+1}\prod_{c=0}^{s}\fall{n-c}{\ell-1},
\qquad 0\le s\le\ell-1,
\]
so that \eqref{eq:main-rec} becomes an \((\ell+1)\)-term recurrence.  In both
cases the optimal bandwidth is \(N=\ell\), as it is for every admissible
\(m\).
\end{remark}

For convenience, we recall explicitly the two conjectures from
\cite{HoShTa} that are relevant here.  To avoid a clash with the notation
\(k,\ell\) used throughout this paper, we denote the order of the
differential operator by \(r\) and a divisor of \(r\) by \(s\).

\medskip\noindent
\textbf{Conjecture 1.9 of \cite{HoShTa}.}
For any irreducible differential operator \(L\) of order \(r\) solving the
generalized algebraic Bochner--Krall problem, the order of the corresponding
difference operator \(\Lambda\) is also \(r\).

\medskip\noindent
\textbf{Conjecture 1.10 of \cite{HoShTa}.}
For any positive integer \(r\), every irreducible differential operator of
order \(r\) solving the generalized algebraic Bochner--Krall problem belongs
to one of the following two types:
\begin{enumerate}
\item
\[
L=\sum_{j=1}^{r} a_j z^{j-1}\partial_z^j+z\partial_z,
\qquad a_j\in\C,\quad a_r\ne0,
\]
which generates the so-called \((r-1)\)-orthogonal polynomials;
\item
\[
L=q'(G)G+z\partial_z,
\]
where \(q\) is a complex polynomial of degree \(r/s\), with no constant
term, \(s\) is any divisor of \(r\), and
\[
G=\left(\sum_{u=0}^{s-1}a_u(z\partial_z)^u\right)\partial_z,
\qquad a_u\in\C,\quad a_{s-1}\ne0.
\]
\end{enumerate}

\begin{proposition}[Relation with Conjectures~1.9 and~1.10 of
\cite{HoShTa}]\label{prop:conjectures}
Assume that \(j=1\) and \(k=\ell-m\) divides \(\ell\), and put
\(\sigma=\ell/k\) and \(D=z\partial_z\).  Define
\begin{equation}\label{eq:G-def}
G=p(D)\partial_z,
\qquad
p(t):=\prod_{r=1}^{\sigma-1}\bigl(t-(rk-1)\bigr),
\qquad
q(t):=\frac{t^k}{k},
\end{equation}
where the empty product is \(1\).  Then \(G\) has differential order
\(\sigma\) and
\begin{equation}\label{eq:type-two}
T=z\partial_z+q'(G)G.
\end{equation}
Since \(p\) is monic of degree \(\sigma-1\), while \(q\) has degree
\(k=\ell/\sigma\) and zero constant term, every admissible operator in
Theorem~\ref{thm:main} has precisely the Type~(2) form in
Conjecture~1.10 of \cite{HoShTa}; this is the construction studied in
\cite[Theorem~2.3]{Horozov}, with the divisor there equal to our \(\sigma\).  Moreover, the associated
difference operator has order \(\ell\).  Hence this family satisfies the
order equality predicted by Conjecture~1.9, without any irreducibility
assumption.
\end{proposition}

\begin{proof}
The commutation relation
\(
\partial_z f(D)=f(D+1)\partial_z
\)
for polynomials \(f\) gives
\[
G^k=\prod_{i=0}^{k-1}p(D+i)\,\partial_z^k.
\]
The roots of the polynomial in front of \(\partial_z^k\) are
\[
rk-1-i,
\qquad 1\le r\le\sigma-1,\quad 0\le i\le k-1.
\]
As a multiset these are precisely
\(\{0,1,\dots,(\sigma-1)k-1\}=\{0,1,\dots,m-1\}\), each occurring once;
when \(m=0\), both multisets are empty.  Consequently
\[
\prod_{i=0}^{k-1}p(D+i)
=D(D-1)\cdots(D-m+1)=z^m\partial_z^m,
\]
and therefore
\[
G^k=z^m\partial_z^{m+k}=z^m\partial_z^\ell.
\]
Since \(q'(t)t=t^k\), identity \eqref{eq:type-two} follows.

Finally, if \(S\) denotes the forward shift in the degree index, then
\eqref{eq:main-rec} corresponds to
\[
\Lambda=S+\sum_{q=1}^{\sigma}\gamma_n^{(q)}S^{1-qk}.
\]
Its highest shift is \(+1\) and its lowest shift is
\(1-\sigma k=1-\ell\).  By the last part of the proof of
Corollary~\ref{cor:classification}, the coefficient of the lowest shift is
nonzero for all sufficiently large \(n\).  Thus the difference operator has
order \(1-(1-\ell)=\ell\).
\end{proof}

\begin{remark}
Proposition~\ref{prop:conjectures} proves the precise structural inclusion
asserted by Conjecture~1.10.  It does not assert that every operator in this
family is irreducible in the sense of \cite{HoShTa}; irreducibility is not
needed for the classification obtained here.
\end{remark}

\medskip \noindent
\emph{Acknowledgements.} Some calculations have been checked and carried out using GPT 5.6 Sol which also  polished the original manuscript.

\end{document}